\newif\ifElsev
\Elsevfalse

\ifElsev
\documentclass[preprint,authoryear]{elsarticle}
\else
\documentclass{article}
\fi
\usepackage{lscape}
\usepackage{url}
\usepackage[all]{xy}
\usepackage{qsymbols}
\DeclareMathAlphabet{\mathpzc}{OT1}{pzc}{m}{it}

\ifElsev 
\else 
\fi

\ifElsev 
\newproof{proof}{Proof} 
\else 
\newenvironment{proof}[1]{\begin{quotation}\noindent\textsf{Proof:} #1}%
{\(\Box\)\end{quotation}}
\fi
\newtheorem{prop}{Proposition}
\newtheorem{corollary}{Corollary}
\newtheorem{conj}{Conjecture}
\newcommand{\ie}{i.e.,~}

\newcommand{\T}{\mathcal{T}}

\newcommand{\F}{\mathcal{F}}
\newcommand{\G}{\mathcal{G}}
\newcommand{\Od}{\mathcal{O}d}
\newcommand{\Sod}{\mathcal{S}od}
\newcommand{\Tod}{\mathcal{T}od}
\newcommand{\Ev}{\mathcal{E}v}
\newcommand{\Sev}{\mathcal{S}ev}
\newcommand{\Fev}{\mathcal{F}ev}
\newcommand{\Fevb}{\overline{\mathcal{F}ev}}
\newcommand{\Fod}{\mathcal{F}od}
\newcommand{\Fodb}{\overline{\mathcal{F}od}}
\newcommand{\SFev}{\mathcal{SF}ev}
\newcommand{\SFevb}{\overline{\mathcal{SF}ev}}
\newcommand{\SFod}{\mathcal{SF}od}
\newcommand{\SFodb}{\overline{\mathcal{SF}od}}
\newcommand{\phib}[1]{\overline{`v}_{#1}}
\newcommand{\psib}[1]{\overline{`j}_{#1}}

\newcommand{\Uu}{U_1}

\newcommand{\Var}[1]{\underline{\mathsf{#1}}}

\usepackage{color} %
\definecolor{darkbrown}{cmyk}{.3,.75,.75,.15}
\definecolor{vertfonce}{rgb}{0,.5,0}

\definecolor{vertfonce}{rgb}{0,.5,0}

\title{On counting untyped lambda terms}

\ifElsev \author[ens]{Pierre Lescanne\corref{cor1}}
\ead{Pierre.Lescanne@ens-lyon.fr}
\address[ens]{ENS de Lyon, University of Lyon, LIP (UMR 5668 CNRS ENS Lyon UCBL INRIA) \\
 46 all\'ee d'Italie, 69364 Lyon, France}
\cortext[cor1]{Corresponding author: \textsf{Pierre.Lescanne@ens-lyon.fr}, {\scriptsize tel: +33 472728683, fax: +33 472728080}}
\else
\author{Pierre Lescanne\\
University of Lyon, ENS de Lyon, \\
 46 all\'ee d'Italie, 69364 Lyon, France}
\fi
\begin{document}

\ifElsev \else 
\date{}
\maketitle 
\pagestyle{plain}
\thispagestyle{empty}
 \fi

 \begin{abstract}
   We present several results on counting untyped lambda terms, \ie on telling
   how many terms belong to such or such class, according to the
   size of the terms and/or to the number of free variables.  
 \end{abstract}

\section{Introduction}
\label{sec:introduction}

This paper presents several results on counting untyped lambda terms, \ie on telling how
many terms belong to such or such class, according to the size of the terms and/or to the
number of free variables.  In addition to the inherent interest of these results from the
mathematical point of view, we expect that a knowledge on the distribution of terms will
improve the implementation of reduction and that results on asymptotic distributions of
terms will give a better insight of the lambda calculus.  For counting more easily lambda
terms we adopt de Bruijn indices that are a well-known coding of bound variables by
natural numbers.  First we give recurrence formulas for the number of terms (and of normal
forms) of size $n$ having $m$ variables.  These recurrence formulas are not familiar in
combinatorics and not amenable to a classical treatment by generating functions.  In a
first part of the paper, we examine the formulas for terms and normal forms when $n$ is
fixed and $m$ varies, which are polynomials. We give the expressions of the first
coefficients of those polynomials since an expression for the generic coefficients seems
out of reach and no regularity appears.  However this shows that these expressions are
clearly connected to Catalan numbers. If we would find an explicit expression for the
last coefficients, this would be an explicit expression for the closed terms.  In the
second part, we give formulas for the generating functions showing the difficulty of a
mathematical treatment.  The results presented here are a milestone in describing
probabilistic properties of lambda terms with answers to questions like: How does a random
lambda term looks like? How does a random normal form look like?  How to generate a random
lambda term (a random normal form)?

\subsection*{Related works}

Previous works on counting lambda terms were by O. Bodini et
al.~\cite{bodini11:_lambd_bound_unary_heigh}, R.~David et
al.~\cite{DBLP:journals/corr/abs-0903-5505} and J.~Wang \cite{Wang05generatingrandom}.
Related works are on counting types and/or counting
tautologies~\cite{DBLP:journals/tcs/Zaionc06,DBLP:conf/csl/FournierGGZ07,%
  DBLP:journals/aml/DavidZ09,%
  DBLP:journals/mlq/FournierGGZ10}. Complexity of rewriting was studied by Choppy et
al.~\cite{ChristineChoppyKS-TCS89}.

\section{Untyped lambda terms with de Bruijn indices}
\label{sec:lambda-terms-with}

\hfill\parbox{10cm}{\it I am dedicating this book to N.~G.~ ``Dick'' de Bruijn, because his
  influence can be felt on every page. Ever since the 1960s he has been my chief mentor,
  the main person who would answer my question when I was stuck on a problem that I had not
  been taught how to solve.
}

\medskip

\rightline{Donald Knuth in preface of~\cite{knuth00:_selec_paper_analy_algor}}
\medskip

The $`l$-calculus~\cite{HBarendregt84} is a logic formalism to describe fonctions, for
instance, the function $f "|->" (x "|->" f (f (x))$, which takes a fonction $f$ and
applies it twice.  For historical reason, this function is written $`l\,f . `l\,x. f (f
x)$, which contains the two variables $f$ and $x$, bounded by $`l$.

In this paper we represent terms by de Bruijn indices~\cite{NGDeBruijn108}, this means that
variables are represented by numbers $\Var{1}, \Var{2}, ..., \Var{m}, ...$, where an
index, for instance~$\Var{k}$, is the number of $`l$'s, above the location of the index and
below the~$`l$ that binds the variable, in a representation of $`l$-terms by trees.  For
instance, the term with variables $`l x. `l y. x\,y$ is represented by the term with de
Bruijn indices $`l `l \Var{2} \Var{1}$. The variable $x$ is bound by the top $`l$.  Above
the occurrence of~$x$, there are two $`l$'s, therefore $x$ is represented by $\Var{2}$ and
from the occurrence of~$y$, we count just the~$`l$ that binds $y$; so $y$ is represented
by $\Var{1}$.  In what follows we will call \emph{terms}, the untyped
terms\footnote{Roughly speaking, typed terms are terms consistent with properties of the domain and
  the codomain of the function they represent.} with de Bruijn
indices.  Let us call $\T_{n,m}$, the set of terms of size $n$, with $m$ de Bruijn
indices, \ie with indices in $\mathcal{I}(m) = \{\Var{1}, \Var{2}, ..., \Var{m}\}$.
A~term in $\T_{n,m}$, is either a de Bruijn index or an abstraction on a term with $m+1$
indices \ie a term in $\T_{n,m+1}$ or an application of a term in $\T_{n,m}$ on a term in
$\T_{n,m}$.  We can write, using $@$ as the application symbol,
\[\T_{n+1,m} = `l \T_{n,m+1} \uplus \biguplus_{k=0}^n\T_{n-k,m} @ \T_{k,m} .\]
Moreover terms of size $1$ are only made of de Bruijn
indices,  therefore 
\begin{eqnarray*}
  \T_{1,m} &=& \mathcal{I}(m).
\end{eqnarray*}
From this we get:
\begin{eqnarray*}
  T_{n+1,m} &=&  T_{n,m+1} + \sum_{k=0}^{n} T_{n-k,m}.T_{k,m} \\
  T_{1,m} & =& m\\
  T_{0,m} &=& 0
\end{eqnarray*}
$\T_{n,0}$ is the set of closed terms (terms with no non bound indices) of size $n$.  Notice that  
\begin{eqnarray*}
  T_{n+1,m} &=&  T_{n,m+1} + \sum_{k=1}^{n-1} T_{n-k,m}.T_{k,m} 
\end{eqnarray*}
Let us illustrate this result by the array of closed terms up to size $5$:

\begin{footnotesize}
  \begin{displaymath}
    \begin{array}{l | l | l}
      \hline
      n & \textrm{terms} & T_{n,0}\\
      \hline
      1& \textrm{non}e & 0 \\
      2& `l \Var{1} & 1 \\
      3& `l `l \Var{1},  `l `l \Var{2}, & 2 \\
      4&  `l `l  `l\Var{1},  `l `l `l\Var{2}, `l `l `l\Var{3}, `l (\Var{1} . \Var{1})& 4 \\
      5&  `l`l `l  `l\Var{1},  `l`l `l  `l\Var{2},  `l`l `l  `l\Var{3},  `l`l `l  `l\Var{4},
      `l`l(\Var{1}.\Var{1}),`l`l(\Var{1}.\Var{2}),`l`l(\Var{2}.\Var{1}),`l`l(\Var{2}.\Var{2}), &13
      \\
      & \qquad \qquad `l (\Var{1} .`l\Var{1}),`l
      (\Var{1} .`l\Var{2}),  `l ((`l\Var{1}) .\Var{1}), `l ((`l\Var{2}) .\Var{1}),`l \Var{1} . `l \Var{1}    & \\ 
          \hline
    \end{array}
  \end{displaymath}
\end{footnotesize}
The equation that defines $T_{n,m}$ allows us to compute it, since it relies on entities
$T_{i,j}$ where either $i<n$ or $j<m$.  Figure~\ref{fig:T} is a table of the first values
of $T_{n,m}$ up to $T_{18,7}$.  We are mostly interested by the sequence of sizes of the
closed terms, namely $T_{n,0}$, in other words the first column of the table.  The values
of $T_{n,0}$ correspond to sequence \textbf{A135501} (see
\url{http://www.research.att.com/~nudges/sequences/A135501}) due to Christophe Raffalli,
which is defined as the \emph{number of closed lambda-terms of size $n$}.  His recurrence formula for those
numbers is more complex.  He considers the values of the double sequence $f_{n,m}$.
$T_{n,m}$ and $f_{n,m}$ coincide for $m=0$, \ie $T_{n,0} = f_{n,0}$.
\begin{eqnarray*}
f_{1,1} &=& 1\\ 
f_{0,m} &=& 0\\  
f_{n,m} &=& 0 \textrm{~ if}~ m>2n-1\\
f_{n,m} &= &f_{n-1,m} + f_{n-1,m+1} + \sum_{p=1}^{n-2} \sum_{c=0}^{m} \sum_{l=0}^{m - c} {c \choose m} {l \choose m-c} f_{p,l+c} f_{n-p-1,m-l}.
\end{eqnarray*}

He adds
\begin{sf}
  \begin{quote}
    The last term is for the application where $c$ is the number of common variables in
    both subterms.  $f_{n,m}$ can be computed only using $f_{n',m'}$ with $n' < n$ and $m'
    \le m + n - n'$.
  \end{quote}
\end{sf}

Notice that he deals only with sequence $T_{n,0}$, whereas we consider the values for any
value of $m$, from which we can expect to extract interesting informations.  The main
interesting statement we can draw from this is that considering lambda terms with explicit
variables or considering lambda terms with de Bruijn indices makes no difference, at least
when no $`b$-reduction is taken into account.  We
feel that considering lambda terms with de Bruijn indices makes the task easier and
produces more results.

\section{Bounding the $T_{n,0}$'s}
\label{sec:bounding-t_0-m}

We can show that the $T_{n,0}$'s are bound by the Motzkin numbers.
More precisely we get the following proposition.
\begin{prop} 
If $M_n$ are the Motzkin numbers,
  \(M_n < T_{n+1,0}.\)
\end{prop}
\begin{proof}{}
  There is a one-to-one correspondance between unary-binary trees and lambda terms of the
  form $`l M$ in which all the indices are~$\Var{1}$.  Hence the results, since Motzkin
  numbers count unary-binary trees.
\end{proof}
We conclude that the asymptotic behavior of the $T_{n,0}$'s are at least $3^n$ since the
Motzkin numbers are asymptotically equivalent to $\sqrt{\frac{3}{4 \pi n^3}}\;3^n$
(\cite{flajolet08:_analy_combin}, Example VI.3).  

\section{The functions $m"|->" T_{n,m}$}
\label{sec:m-Tnm}

Due to properties of the generating function (see Section~\ref{sec:gener-funct-terms}) we
are not able to give a simple expression for the function $n"|->" T_{n,m}$, so we focus on
the function $m "|->" T_{n,m}$.  These functions are
polynomials $P^T_n$, defined recursively as follows:
\pagebreak[2]
\begin{eqnarray}
  P^T_0 (m) &=& 0 \label{eq:a}\\
  P^T_1(m) &=& m\\
  P^T_{n+1}(m) &=& P^T_{n}(m+1) + \sum_{k=1}^{n-1} P^T_{k}(m)\, P^T_{n-k}(m). \label{eq:c}
\end{eqnarray}
See Figure~\ref{fig:Pn} for the first $18$ polynomials.  
 The table below gives
the coefficients of the polynomials $P^T_{n}$ up to~$13$.

\begin{tiny}
  \begin{displaymath}
    \begin{array}{r|r| r | r | r | r| r | r| r | r}
      n\backslash m^i & m^8 & m^7 & m^6 & m^5 & m^4 & m^3 & m^2 & m & 1\\
      \hline\hline
      1 &&&&&&&& 1 & 0\\
      \hline
     \mathit{ 2} &&&&&&&& \mathit{1} & \mathit{1} \\
      \hline
      3 && &&&&&1 & 1 &2 \\
      \hline
      \mathit{4} &&&&&&&\mathit{3} &\mathit{5} &\mathit{4}  \\
      \hline
      5 && &&&&2& 6&17&13\\
      \hline
      \mathit{6} && &&&&\mathit{10} & \mathit{26}&\mathit{49} & \mathit{42}\\
      \hline
      7 &&&&& 5 & 30 & 111&179&139\\
      \hline
      \mathit{8}  &&&&& \mathit{35}  & \mathit{134} & \mathit{405}&\mathit{683} & \mathit{506}\\
      \hline
      9 &&&& 14 &140 &652 &1658 &2629 &1915\\
      \hline
      \mathit{10} &&&& \mathit{126} &\mathit{676} &\mathit{2812} &\mathit{7122} &\mathit{10725} &\mathit{7558}\\
      \hline
      11&&& 42 &630 &3610 &12760 &30783 &45195 &31092\\
      \hline
      \mathit{12} &&& \mathit{462} &\mathit{3334} &\mathit{17670} &\mathit{60240} &\mathit{138033} &\mathit{196355} &\mathit{132170}\\
      \hline
      13 && 132 &2772 &19218 &87850 &285982 &635178 &880379 &580466\\
      \hline
      \mathit{14} && \mathit{1716} &\mathit{16108} &\mathit{104034} &\mathit{449290} &\mathit{1390246} &\mathit{2991438} &\mathit{4052459} & \mathit{2624545}\\ 
      \hline
      15 & 429 &12012 &99386 &560854 &2308173&6895122 &14436365 &19144575 &12190623\\
      \hline
      \mathit{16} & \mathit{6435} &\mathit{76444}&\mathit{584878}&\mathit{3076878} &\mathit{12039895}&\mathit{34815210} &\mathit{71170791}&\mathit{92631835} &\mathit{58083923}\\
      \hline
    \end{array}
  \end{displaymath}
\end{tiny}
The degrees of those polynomials
increase two by two and we can describe their leading coefficients,  their second leading
coefficients and the third leading coefficients of the odd polynomials.

\begin{prop}\label{prop:deg}
$deg(P^T_{2p-1}) = deg(P^T_{2p}) = p.$
\end{prop}
\begin{proof}{}
  This is true for $P^T_1 = m$ and $P^T_2 = m+1$ which have degree~$1$.  Assume the
  property true up to $p$.  Note that all the coefficients of the $P_n^T$'s are positive. In
\[P^T_{n}(m+1) + \sum_{k=1}^{n-1} P^T_{k}\, P^T_{n-k},\] the degree of $P^T_{n+1}(m)$ comes
from the $ P^T_{k}\, P^T_{n-k}$'s.  Indeed, par induction the degree of $P_n^T(m+1)$ is
$n\div 2$ which is less than $(n+1)\div 2$, therefore we can consider that $P_n^T(m+1)$
does not contribute to the degree of $P^T_{n+1}(m)$. 
Consider the degree of $ P^T_{k}\, P^T_{n-k}$
according to the parity of $n$ and $k$.
\begin{description}
\item[$n=2p+1$ and $k=2h-1$.]  In this case, $p\ge h\ge 1$ and the degree of $P^T_{2h-1}$ is $h$ and the
  degree of $P^T_{2p+1-2h +1}$ is $p-h+1$, hence the degree of $ P^T_{2h-1}\,
  P^T_{2p+1-2h+1}$ is $p+1$.
\item[$n=2p+1$ and $k=2h$.] In this case, $p\ge h \ge 1$ and the degree of $P^T_{2h}$ is $h$ and the
  degree of $P^T_{2p+1-2h}$ is $p-h+1$, hence the degree of $ P^T_{2h}\,
  P^T_{2p+1-2h}$ is $p+1$.
\item[$n=2p$ and $k=2h-1$.] In this case, $p+1\ge h \ge 1$ and the degree of
  $P^T_{2h-1}$ is $h$ and the degree of $P^T_{2p-2h +1}$ is $p-h+1$, hence the degree
  of $ P^T_{2h-1}\, P^T_{2p+2-2h+1}$ is $p+1$.
\item[$n=2p$ and $k=2h$.] In this case, $p+1\ge h \ge 1$ and the degree of
  $P^T_{2h}$ is $h$ and the degree of $P^T_{2p-2h}$ is $p-h$, hence the degree
  of $ P^T_{2h}\, P^T_{2p-2h}$ is $p$.  These products $ P^T_{2h}\, P^T_{2p-2h}$  do not
  contribute to the degree of $ P^T_{2p+1}$.
\end{description}

\end{proof}
In what follows, for short, we write $`q_{2q+1}$ and $`q_{2q}$ the leading coefficients of
$P^T_{2q+1}$ and $P^T_{2q}$, $`t_{2q+1}$ and $`t_{2q}$ the second leading coefficients of
$P^T_{2q+1}$ and $P^T_{2q}$, and $`d_{2q+1}$ the third leading coefficients
of$P^T_{2q+1}$.  We also write, as usual, $C_n$ the $n^{th}$ Catalan number.

We define five generating functions.
\begin{displaymath}
  \Od(z) = \sum_{i=0}^{\infty} `q_{2i+1} z^i \qquad
  \Ev(z) = \sum_{i=0}^{\infty} `q_{2i} z^i
 \end{displaymath}
  \begin{displaymath}
   \Sod(z) = \sum_{i=0}^{\infty} `t_{2i+1}z^i  \qquad
    \Sev(z) = \sum_{i=0}^{\infty}`t_{2i}z^i
  \end{displaymath}
  \begin{displaymath}
\Tod(z) = \sum_{i=0}^{\infty} `d_{2i+1}z^i.
\end{displaymath}

\begin{prop}\label{prop:lead_co_P_odd}
  The leading coefficients of $P^T_{2q+1}$ are $\frac{1}{q+1}{2q \choose q}$, \ie the
  Catalan numbers $C_q$.
\end{prop}
\begin{proof}{}
  From Equation~(\ref{eq:c}) and the last two steps of the proof of
  Proposition~\ref{prop:deg}, we deduce the following relation :
\begin{eqnarray*}
     `q_{2q+1} &=& \sum_{h=0}^{q-1} `q_{2h+1}\, `q_{2q- 2h-1} \qquad \textrm{for~} q\ge 1\\
     `q_1 &=&1.
\end{eqnarray*}
which says that the leading coefficient of an odd  polynomial comes only from the leading
coefficients in the products of odd polynomials.
We get:
\begin{displaymath}
  \Od(z) = 1 + z\,\Od(z)^2.
\end{displaymath}
which shows that
\begin{displaymath}
\Od(z) = \frac{1-\sqrt{1-4z}}{2z}.
\end{displaymath}
and $\Od(z) = C(z)$, the generating function of the Catalan numbers.
\end{proof}
\begin{prop}
  The leading coefficients of $P^T_{2q}$ are ${ 2q-1 \choose q}$, for $q\ge 1$.
\end{prop}
\begin{proof}{}
Without lost of generality, we  assume that $`q_0=0$.
From Equation~(\ref{eq:c}), we get, for $q\ge 1$,
  \begin{eqnarray*}
    `q_{2q+2} &=& `q_{2q+1} + \sum_{k=0}^{2q+1} `q_{k}\, `q_{2q+1-k} \\
    &=& `q_{2q+1} + \sum_{h=0}^{q}`q_{2h} \,`q_{2q+1-2h} + \sum_{h=0}^{q} `q_{2h+1}\, `q_{2q - 2h} \\
    &=& `q_{2q+1} + 2\,\sum_{h=0}^{q} `q_{2h}\, `q_{2q - 2h+1}.
   \end{eqnarray*}
which says that the leading coefficient of an even polynomial comes from the leading
coefficient of the preceding odd polynomial and of the products of the leading coefficients of
the products of the smaller polynomials. 
We get:
\begin{displaymath}
  \Ev(z) = z\Od(z) + 2\,z \Ev(z)\Od(z),
\end{displaymath}
hence
\begin{displaymath}
  \Ev(z) = \frac{z \Od(z)}{1-2z \Od(z)} = \frac{1-\sqrt{1-4z}}{2\sqrt{1-4z}} =
  \frac{\sqrt{1-4z}}{2(1-4z)} - \frac{1}{2}
\end{displaymath}
which is the generating function of  the sequence ${ 2q+1 \choose q}$.
\end{proof}
\begin{prop}
  The second leading coefficients of $P^T_{2q+1}$ are $(2q-1){ 2(q-1) \choose q-1}$.
\end{prop}
\begin{proof}{}
 From the proof of Proposition~\ref{prop:deg}, we see that the
  monomial of second highest degree of $P_{2q+1}$ is made as the sum:
  \begin{itemize}
  \item of the monomial of highest degree of $P_{2q}$,
  \item of the products of the
    monomials of highest degree from the $P_i$'s with even indices and
  \item the products of
    monomials of highest degree with monomials of second highest degree from the $P_i$'s
    with odd indices.
  \end{itemize}
We get for $q\ge 1$:
  \begin{eqnarray*}
    `t_{2q+1} &=& `q_{2q} + \sum_{h=0}^{q} `q_{2h}`q_{2q-2h} + \sum_{h=0}^{q-1} `q_{2h+1}
    `t_{2q-2h -1} + \sum_{h=0}^{q-1} `t_{2h+1} `q_{2q-2h-1}.
  \end{eqnarray*}
We notice that $`t_1=0$.  Therefore we get:
\begin{eqnarray*}
  \Sod(z) &=& \Ev(z) + \Ev(z)^2 + 2z \Od(z)\,\Sod(z).
\end{eqnarray*}
Then
\begin{displaymath}
   \Sod(z)  = \frac{\Ev(z)+\Ev(z)^2}{1-2z\Od(z)}\\
   =  \frac{z}{\sqrt{1-4z}(1-4z)} = \frac{z\sqrt{1-4z}}{(1-4z)^2}
\end{displaymath}
which is the generating function of $(2q-1){ 2(q-1) \choose q-1}$.
\end{proof}
\begin{prop}
  The second leading coefficients of $P^T_{2q}$ are $`t_0=0$, $ `t_1=1$, $`t_2= 5$ and for  $q\ge 3$,
 \begin{displaymath}
 `t_{2q} \ =\ 4^{q-1} + \frac{2 (2q - 5)  (2q - 3)  (2q -1)}{3(q-2)}  {2(q-3) \choose q-3}.
\end{displaymath}
\end{prop}
\begin{proof}{}
From Equation~(\ref{eq:c}),we get
\begin{displaymath}
  \left\{
    \begin{array}{lcl}
      `t_{2q+2} &=& (q+1) `q_{2q+1} + `t_{2q+1} \\&&+ 2 \sum_{i=1}^q`q_{2i-1}`t_{2q-2i+2} 
      + 2 \sum_{i=1}^q`q_{2i}`t_{2q-2i+1}\\
      `t_0 &=&0
    \end{array}
\right.
\end{displaymath}
The second leading coefficient of an even polynomial $P_{2m+2}^T$ is made of four components:
\begin{itemize}
\item the coefficient of degree $q$ in $`q_{2q+1}(m+1)^{q+1}$, namely ${(q+1) `q_{2q+1}}$,
\item the coefficient of degree $q$ in $`t_{2q+1}(m+1)^q$, namely $`t_{2q+1}$,
\item the sum of the products of the leading coefficients of the odd polynomials and the
  second leading coefficients of the even polynomials (this occurs twice, once in product
  $P_{2i-1}\,P_{2q-2i+2}$ and once in product $P_{2i}\,P_{2q-2i+1}$),
\item the sum of the products of the leading coefficients of the even polynomials and the
  second leading coefficients of the odd polynomials (twice).
\end{itemize}

From the above induction, $\Sev$ fullfils the following functional equation:
\begin{displaymath}
  \Sev(z) = z \Od(z) + z^2\Od'(z) + z \Sod(z) + 2z \Od(z)\Sev(z) + 2z\Ev(z)\Sod(z).
\end{displaymath}
Therefore
\begin{eqnarray*}
  \Sev(z) &=& \frac{z \Od(z) + z^2\Od'(z) + z \Sod(z) + 2z\Ev(z)\Sod(z)}{\sqrt{1-4z}}\\
&=& \frac{(1-\sqrt{1-4z})}{2\sqrt{1-4z}} \\
&&+ \frac{z}{1-4 z}-\frac{1-\left(\sqrt{1-4 z}\right)}{2 \sqrt{1-4z}}\\
&&+ \frac{z^2}{(1-4z)^2}\\
&& + \frac{z(1-\sqrt{1-4z})}{(1-4z)^2\sqrt{1-4z}}\\
&=& \frac{z}{1-4 z} + \frac{z^2}{(1-4z)^2} +
\frac{z^2(1-\sqrt{1-4z})}{(1-4z)^2\sqrt{1-4z}}\\
&=& \sum_{q=1}^{\infty}4^{q-1} z^q + \sum_{q=2} ^{\infty}(q-1) 4^{q-2} z^q + \sum_{q=3}^{\infty} 2a_{q-3}z^q
\end{eqnarray*}
where $(a_n)_{n\in\mathbf{N}}$ is sequence \textsf{A029887} of  the \emph{On-Line Encyclopedia of
  Integer Sequences} whose value is:
\[\frac{(2n+1)(2n+3)(2n+5)}{3}C_n-(n+2)2^{2n+1}.\]
  Hence 
\begin{eqnarray*}
  \Sev(z) &=& \sum_{q=1}^{\infty}4^{q-1} z^q + \sum_{q=3}^{\infty}  \frac{2 (2q - 5)(2q-3)(2q -1)}{3}  C_{q-3}z^q\\
  &=& \frac{z}{1-4z} + \frac{z^2}{(1-4z)^2\sqrt{1-4z}}.
\end{eqnarray*}
\end{proof}
\begin{prop}
  The third leading coefficients of $P^T_{2q+1}$ are 
  \[q\;2^{2q-1} +\frac{q(q-1)(q-2)}{120}{2q\choose q}+ \frac{(q+1)q(q-1)}{120}{2(q+1)
    \choose q+1}.\]
\end{prop}
 \begin{proof}{}
    Since $deg(P^T_{2n}) = deg(P^T_{2n+1} )-1$, the third coefficient is the sum of seven
  items:
  \begin{itemize}
  \item the second coefficient of $`q_{2q}\,(m+1)^q$, namely $q`q_{2q}$,
  \item the first coefficient of $(m+1)^{q-1}$, namely $`t_{2q}$,
  \item the sum of products of leading coefficients and second leading coefficients for
    even polynomials (twice),
  \item the sum of leading coefficients and third leading coefficients for odd polynomials
    (twice),
  \item the sum of second leading coefficients with second leading coefficients.
  \end{itemize}
The formula for $`d_{2q+1}$ is:
\begin{eqnarray*}
  `d_{2q+1} &=& q\,`q_{2q} + `t_{2q} + \sum_{i=0}^q `t_{2i} `q_{2q-2i} + \sum_{i=0}^q
  `q_{2i}`t_{2q-2i} + \\
  &&\sum_{i=0}^{q-1} `q_{2i+1} `d_{2q-2i-1}  +  
  \sum_{i=0}^{q-1} `d_{2i+1} `q_{2q-2i-1} + \sum_{i=0}^{q-1} `t_{2i+1} `t_{2q-2i-1} ,
\end{eqnarray*}
which gives the following equation on generating functions:
\begin{eqnarray*}
  \Tod(z) &=& z\,\Ev'(z) + \Sev(z) + 2z\Ev(z)\Sev(z) +\\
  && 2z \,\Od(z)\Tod(z)+ z\,\Sev(z)^2.
\end{eqnarray*}
which yields:
\begin{eqnarray*}
  \Tod(z) &=& \frac{z\,\Ev'(z) + \Sev(z) + 2z\Ev(z)\Sev(z) + z\,\Sev(z)^2}{1-2z\,\Od(z)}\\
&=& \frac{1}{\sqrt{1-4z}}\,\left(\frac{z}{(1-4z)\sqrt{1-4z}} + \right.\\
&& \qquad \qquad \frac{z}{1-4z}+ \frac{z^2}{(1-4z)^2\sqrt{1-4z}}+\\
&&\qquad \qquad \frac{1-\sqrt{1-4z}}{\sqrt{1-4z}} \left(\frac{z}{1-4z} + \frac{z^2}{(1-4z)^2\sqrt{1-4z}}\right)+\\
&&\left.\qquad\qquad \frac{z^3}{(1-4z)^3}\right)\\
&=& \frac{2z}{(1-4z)^2} + \frac{z^2+z^3}{(1-4z)^3\sqrt{1-4z}}.
\end{eqnarray*}
The first part corresponds to sequence \textsf{A002699} which expression is
$q\,2^{2q-1}$. $1/(1-4z)^3\sqrt{1-4z}$ corresponds to sequence \textsf{A144395}. Therefore
the second part yields the expression %
\[\frac{q(q-1)(q-2)}{120}{2q\choose q}+ \frac{(q+1)q(q-1)}{120}{2(q+1)\choose q+1}.\]
  \end{proof}
Hence typically if we pose
\begin{eqnarray*}
  `t_{2q}&= & 4^{q-1} +  \frac{2 (2q - 5)  (2q - 3)  (2q -1)}{3}  C_{q-3} \\
  `d_{2q+1} &=& q\;2^{2q-1} +\frac{(q+1)q(q-1)(q-2)}{120}C_q+ \frac{(q+2)(q+1)q(q-1)}{120}C_{q+1}
\end{eqnarray*}
we have in general:
  \begin{eqnarray*}
P^T_{2q}(m) &=& (2q-1) C_{q-1} m^q + `t_{2q} m^{q-1} + \ldots + T_{2q,0}\\
P^T_{2q+1} (m) &=& C_q m^{q+1} + \frac{2q(2q-1)}{2} C_{q-1} m^{q} + `d_{2q+1} m^{q-1} + \ldots + T_{2q+1,0}
\end{eqnarray*}
showing the prominent role of Catalan numbers.  The relations for the other coefficients
are more convoluted\footnote{Like $`t_{2q}$ and $`d_{2q+1}$, they correspond to non
  studied sequences according to the \emph{On-Line Encyclopedia of Integer Sequences}.}
and have not been computed.

It should be interesting to study the connection with the derivatives of the generating
function $C(z)$ of the Catalan numbers~\cite{lang02:_polyn_cataly}. 
\section{Normal forms}
\label{sec:normal-forms}

Lest us call $\F_m$ the set of normal forms with $\{\Var{1},.., \Var{m}\}$ de Bruijn indices and $\G_m$ the sets
of normal forms with no head $`l$ and also  de Bruijn indices in $\{\Var{1},.., \Var{m}\}$.  The combinatorial structure equations are
\begin{eqnarray*}
  \G_m &=&  \mathcal{I}(m) \uplus \G_m @ \F_m\\
  \F_m &=& `l\,\F_{m+1} \uplus \G_m
\end{eqnarray*}
Let $G_{n,m}$ be the number of normal forms of size $n$ with no head $`l$ and with de
Bruijn indices in $\mathcal{I}(m)$ and let $F_{n,m}$ be the number of normal forms of size $n$
with de Bruijn indices in $\mathcal{I}(m)$.  The relations between $G_{n,m}$ and $F_{n,m}$
are
\begin{eqnarray*}
  G_{0,m} &=& 0\\
  G_{1,m} &=& m\\
  G_{n+1,m} &=& \sum_{k=0}^{n}G_{n-k,m} F_{k,m}\\\\
  F_{0,m} &=& 0\\
  F_{1,m} &=& m \quad = \quad G_{1,m} \\
  F_{n+1,m}&=& F_{n,m+1} + G_{n+1,m}
\end{eqnarray*}
whereas the relations between generating functions are
\begin{eqnarray*}
  G_m(z) &=& m\,z + z \,G_m(z)\,F_m(z)\\
  F_m(z) &=& z\, F_{m+1}(z) + G_m(z).
\end{eqnarray*}
The coefficients $F_{n,m}$ are given in Figure~\ref{fig:F}.  

\subsection*{The functions $m"|->" F_{n,m}$}
Like for $m "|->" T_{n,m}$, the functions $m "|->" F_{n,m}$ are polynomials of degree $n+1
\div 2$, which we write $P^{NF}_n$ and which are given in Figure~\ref{fig:polyF}.  The
coefficients of polynomials~$P^{NF}_n$ enjoy properties somewhat similar to those proved
for polynomials~$P^T_n$.  In this section, we write $P_n(m)$ the polynomial $P_n^{NF}(m)$,
$Q_n(m)$ the polynomial associated with $G_{n,m}$, $`v_n$ the leading coefficient of $P_n$,
$\phib{n}$ the leading coefficient of $Q_n$, $`j_n$ the second leading coefficient of
$P_n$ and $\psib{n}$ the second leading coefficient of of $Q_n$.  We have the equations
\begin{eqnarray}
  \label{eq:d}
  P_{n+1}(m) &=& P_n(m+1) + Q_{n+1}(m)\\\label{eq:e}
  Q_{n+1}(m) &=&\sum_{k=0}^nQ_{n-k}(m)P_k(m)
\end{eqnarray}
\begin{prop}\label{propNF:deg}
$deg(P_{2p-1}) = deg(P_{2p}) = deg(Q_{2p-1} )= deg(Q_{2p}) = p$.
\end{prop}
\begin{proof}{}
Here also the coefficients are positive.  The degree of $P_n$ is the degree of $Q_n$ by
(\ref{eq:d}). One notices that $deg~P_0 = deg~Q_0 = 0$ and $deg~P_1 = deg~Q_1=1$.  The
general step can be mimicked from this of Prop~\ref{prop:deg}.
\end{proof}
We define eight generating
functions:
\begin{displaymath}
  \Fev(z) \ = \ \sum_{i=0}^{\infty} `v_{2i} z^i \qquad
  \Fevb(z) \ = \ \sum_{i=0}^{\infty} \phib{2i}z^i
\end{displaymath}
\begin{displaymath}
  \Fod(z) \ = \  \sum_{i=0}^{\infty} `v_{2i+1} z^i \qquad
  \Fodb(z) \ =\ \sum_{i=0}^{\infty} \phib{2i+1}z^i
\end{displaymath}
\begin{displaymath}
  \SFev(z) \ = \ \sum_{i=0}^{\infty} `j_{2i} z^i \qquad
  \SFevb(z) \ = \ \sum_{i=0}^{\infty} \psib{2i}z^i
\end{displaymath}
 \begin{displaymath}
  \SFod(z) \ = \  \sum_{i=0}^{\infty} `j_{2i+1} z^i \qquad
  \SFodb(z) \ =\ \sum_{i=0}^{\infty} \psib{2i+1}z^i
\end{displaymath}
 \begin{prop}
    The leading coefficients of $P^{NF}_{2q+1}$ are Catalan numbers.
  \end{prop}
  \begin{proof}{}
    We see easily that $`v_{2q+1} = \phib{2q+1}$ by (\ref{eq:d}).  By (\ref{eq:e}), we see
    that
    \begin{eqnarray*}
      `v_{2q+1} &=& \sum_{h=0}^{q-1} \phib{2q+1}`v_{2q-2h-}1\\
      `v_1 &=& \phib{1} = 1.
    \end{eqnarray*}
    Hence the result $\Fod(z) = \Fodb(z) = C(z)$ (see proof of
    Proposition~\ref{prop:lead_co_P_odd}).
  \end{proof}
  \begin{prop}\label{prop:Fev}
    The leading coefficients of the $P^{NF}_n$'s for $n$ even, are $P^{NF}_0 =0$, $P^{NF}_2
    =1$ and $P^{NF}_{2q+4} = 2{2q+1 \choose q}$, 
    \ie $P^{NF}_{2q+4} = 2P^T_{2q+2}$.
  \end{prop}
  \begin{proof}{}
    From Equations~(\ref{eq:d}) and (\ref{eq:e}),we get:
    \begin{eqnarray*}
      `v_{2(q+1)} &=& `v_{2q+1} + \phib{2(q+1)}\\
      \phib{2(q+1)} &=& \sum_{i=0}^q \phib{2q+1-2i} `v_{2i} + \sum_{i=0}^q \phib{2q-2i} `v_{2i+1} 
    \end{eqnarray*}
Hence
\begin{eqnarray*}
  \Fev(z) &=& z \Fod(z) + \Fevb(z)\\
\Fevb(z) &=& z \Fodb(z)\Fev(z) + z \Fevb(z) \Fod(z)
\end{eqnarray*}
from which we get 
\begin{eqnarray} \label{eq:Fevb}
  \Fevb(z) \ = \ \frac{z\Fod(z)\Fev(z)}{1-z \Fod(z)}
\end{eqnarray}
then
\begin{displaymath}
  \Fev(z) \ = \ z\Fod(z) + \frac{z\Fod(z)\Fev(z)}{1-z \Fod(z)}
\end{displaymath}
and
\begin{displaymath}
  \Fev(z) - z \,\Fod(z)\Fev(z) \ = \ z\, \Fod(z) - z^2\Fod(z)^2 + z\, \Fod(z)\Fev(z)
\end{displaymath}
and
\begin{eqnarray*}
  \Fev(z) &=&\frac{z \Fod(z) - z^2\Fod(z)^2}{1-2z \Fod(z)}\\
    &=& \frac{z}{\sqrt{1-4z}} \ = \ \frac{z\sqrt{1-4z}}{1-4z} \\
    &=& \frac{z}{1-2zC(z)}.
\end{eqnarray*}
Hence $\Fev(z)$ is the generating function of the sequence $`v_0=0$, $`v_2=1$ and
$`v_{2q+4} = 2{2q+1 \choose q}$.
  \end{proof}
  \begin{corollary}
    $\Fevb(z) = \frac{1-2z - \sqrt{1-4z}}{2\sqrt{1-4z}}$
  \end{corollary}
  \begin{proof}{} 
    \begin{eqnarray*}
      \Fevb(z) &=& \Fev(z) - z C(z)\\
      &=& \frac{z}{\sqrt{1-4z}} - \frac{1-\sqrt{1-4z}}{2} =  z^2 C'(z).
    \end{eqnarray*}
  \end{proof}
 \begin{prop}
The second leading coefficients of  the $P^{NF}_{2q+1}$'s are $`j_0=0$, ${`j_3=1}$ and
$`j_{2q+5} = (q+3) {2q+1 \choose q}$.
\end{prop}
\begin{proof}{}
From the proof of Proposition~\ref{propNF:deg},
\begin{eqnarray*}
  `j_{2q+1} &=& `v_{2q} + \psib{2q+1}\\
  \psib{1}&=&0\\
\psib{2q+3} &=& \sum_{i=0}^{q+1} \phib{2i} `v_{2q-2i} + \sum_{i=0}^{q} \psib{2i+1}
`v_{2q-2i+1} + \sum_{i=0}^{q} `j_{2i+1} \phib{2q-2i+1},
\end{eqnarray*}
from which we get
\begin{eqnarray*}
  \SFod(z) &=& \Fev(z)+\SFodb(z)\\
\SFodb(z) &=& \Fevb(z)\Fev(z) + z\, \SFodb(z) \Fod(z) + z\, \SFod(z) \Fodb(z).
\end{eqnarray*}
Then we get
\begin{eqnarray*}
  \SFodb(z) (1-z\Fod(z))&=& \Fev(z) \Fevb(z) + z \SFod(z) \Fodb(z).
\end{eqnarray*}
We know that $1-z\Fod(z) = 1-zC(z) = 1/ C(z)$, then
\begin{eqnarray*}
  \SFodb(z) &=& \frac{z}{\sqrt{1-4z}}\ z^2C'(z) C(z) + z\, \SFod(z) C(z)^2
\end{eqnarray*}
and
\begin{eqnarray*}
  \SFod(z) &=& \Fev(z) + \frac{z^3\,C(z)C'(z)}{\sqrt{1-4z}} + z\SFod(z) C(z)^2.
\end{eqnarray*}
We know $1 - z\,C(z)^2 = C(z) \sqrt{1-4z}$, then
\begin{eqnarray*}
  \SFod(z) &=& \left(\frac{z}{\sqrt{1-4z}} + \frac{z^3C(z)C'(z)}{\sqrt{1-4z}}\right)\,
\frac{1}{C(z)\,\sqrt{1-4z}}\\
&=& \frac{z}{C(z)\,(1-4z)} + \frac{z^3\,C'(z)}{1-4z}\\
&=& \frac{z^2}{(1-4z)\sqrt{1-4z}} + \frac{z}{\sqrt{1-4z}}.
\end{eqnarray*}
which is the generating function of the sequence $0, 1$ followed by  $(q+3) {2q+1 \choose q}$.
\end{proof}
\begin{corollary}
  \begin{math}
  \SFodb(z) = \frac{z^2}{(1-4z)\sqrt{1-4z}}
\end{math}
\end{corollary}
\begin{proof}{}
  \[\SFodb(z) = \SFod(z) -\Fev(z) = \frac{z^2}{(1-4z)\sqrt{1-4z}}.\]
Notice that $\SFodb(z) = z\Sod(z)$.
\end{proof}
\begin{prop}
  The second leading coefficients of  the $P^{NF}_{2q}$'s are $`j_0=0$, ${`j_{2} = 1}$,
  $`j_6 = 15$ and for $q\ge 4$
\begin{eqnarray*}
  `j_{2q} &=& {2q-3 \choose q-2} + 2^{2q-3} + (q-2){2q-2 \choose q-2} + \\
  && 2{2q-5 \choose q-3} + \frac{(q-3)(q-2)}{3}{2q-5 \choose q-3}.
\end{eqnarray*}
\end{prop}
\begin{proof}{}
  We have
  \begin{eqnarray*}
    `j_{2q+2} &=& (q+1)`v_{2q+1} + `j_{2q+1} + \psib{2q+2}\\
    \psib{2q+2} &=& \sum_{i=1}^{q}`j_{2i-1}\phib{2q-2i+2} + \sum_{i=1}^q
    `v_{2i-1}\psib{2q-2i+2} +\\
    && \sum_{i=1}^q \phib{2i-1} `j_{2q-2i+2} + \sum_{i=1}^q\psib{2i-1} `v_{2q-2i+2}.
  \end{eqnarray*}
This gives the equations on generic functions.
\begin{eqnarray*}
\SFev(z) &=& z \Fod(z) + z^2\Fod'(z) + z \SFod(z) + \SFevb(z)\\
  \SFevb(z) &=& z \SFod(z)\Fevb(z) + z \Fod(z) \SFevb(z) +\\
  && z \SFev(z) \Fodb(z) + z \Fev(z)\SFodb(z).
\end{eqnarray*}
Hence
\begin{eqnarray*}
  \SFevb(z) &=& \frac{z \SFod(z)\Fevb(z) + z \SFev(z) \Fodb(z) + z \Fev(z)\SFodb(z)}{1-zC(z)}
\end{eqnarray*}
which yields
\begin{eqnarray*}
  \SFev(z) & = &  \Fod(z) + z^2\Fod'(z) + z \SFod(z) + \\
  && C(z) ( z \SFod(z)\Fevb(z) + z \Fev(z)\SFodb(z)) \\
  && z C(z) \SFev(z) \Fodb(z).
\end{eqnarray*}
and
\begin{eqnarray*}
  SFev(z) &=& \frac{\Fod(z) + z^2\Fod'(z) + z \SFod(z)}{1-z C(z)^2}+ \\
  && \frac{z C(z) \SFod(z)\Fevb(z) + z    C(z)\Fev(z)\SFodb(z)}{1-z C(z)^2}\\
 &=& \frac{z}{\sqrt{1-4z}} + \frac{z^2 C'(z)}{C(z)\sqrt{1-4z}} + \\
 && \frac{z}{C(z)\sqrt{1-4z}}\left(\frac{z^2}{(1-4z)\sqrt{1-4z}} +  \frac{z}{\sqrt{1-4z}}\right) + \\
 && \left( \frac{z}{\sqrt{1-4z}} - \frac{1-\sqrt{1-4z}}{2}\right)\left( \frac{z^2}{(1-4z)^2} + \frac{z^2}{1-4z}\right) + \\
&& \frac{z^4}{(1-4z)^2\sqrt{1-4z}}.
\end{eqnarray*}
Notice that
\begin{eqnarray*}
  \frac{z^2 C'(z)}{C(z)\sqrt{1-4z}} &=&\frac{z}{2(1-4z)} - \frac{z}{2\sqrt{1-4z}}.
      \end{eqnarray*}
and
  \begin{eqnarray*}
 \frac{z}{C(z)\sqrt{1-4z}}\left(\frac{z^2}{(1-4z)\sqrt{1-4z}} +
   \frac{z}{\sqrt{1-4z}}\right)  + \\
\left( \frac{z}{\sqrt{1-4z}} - \frac{1-\sqrt{1-4z}}{2}\right)\left( \frac{z^2}{(1-4z)^2} +
  \frac{z^2}{1-4z}\right)  &=& \frac{2z^3 }{\sqrt{1-4z}(1-4z)} +
\frac{z^2}{\sqrt{1-4z}}+\\
&&\frac{z^4}{\sqrt{1-4z}(1-4z)^2} 
\end{eqnarray*}
Hence
\begin{eqnarray*}
  \SFev(z) &=& \frac{z}{2\sqrt{1-4z}} + \frac{z}{2(1-4z)} +\\ && \frac{2z^3 }{\sqrt{1-4z}(1-4z)} +
\frac{z^2}{\sqrt{1-4z}}+\frac{2z^4}{\sqrt{1-4z}(1-4z)^2} 
\end{eqnarray*}

We summarize the result in the following table.

\medskip

\begin{center}
  \begin{tabular}{l| l | l |l}
    \textit{gen. fonct.} & \textit{coefficients} & \textit{up to}& \textit{why?} \\\hline\hline
    $\frac{z}{2\sqrt{1-4z}}$ & ${2q-3 \choose q-2}$ & $q\ge 2$ & Proposition~\ref{prop:Fev} \\\hline
    $\frac{z}{2(1-4z)}$ &  $2^{2q-3}$ &  $q\ge 2$ \\\hline
    $\frac{2z^3 }{\sqrt{1-4z}(1-4z)}$ & $(q-2){2q-2 \choose q-2}$& $q \ge 2$\\\hline
    $\frac{z^2}{\sqrt{1-4z}}$ & $2{2q-5 \choose q-3}$ & $q\ge 3$\\\hline
    $\frac{2z^4}{\sqrt{1-4z}(1-4z)^2}$ &  $\frac{(q-3)(q-2)}{3}{2q-5 \choose q-3}$& $q\ge 4$& \textsf{A002802}\\\hline
  \end{tabular}
\end{center}

\medskip

Hence we have for $q\ge 4$:
\begin{eqnarray*}
  `j_{q} &=& {2q-3 \choose q-2} + 2^{2q-3} + (q-2){2q-2 \choose q-2} + \\
  && 2{2q-5 \choose q-3} + \frac{(q-3)(q-2)}{3}{2q-5 \choose q-3}.
\end{eqnarray*}
\end{proof}
Recall what we have computed for\textit{ plain terms}:
\begin{displaymath}
\begin{array}{|l|l|l|l|l|}
\hline\hline
\textsl{coefficients}&\multicolumn{2}{c|}{\textsl{generating functions}}& \textsl{values} & \textsl{equivalents}\\
\hline\hline
P^T_{2q+1,q+1}&\Od(z)& \frac{1-\sqrt{1-4z}}{2z}&C_q& 4^q\sqrt{\frac{1}{\pi q^3}}\\ \hline
P^T_{2q+1,q}&\Sod(z) & \frac{z}{(1-4z)\sqrt{1-4z}}&
(2q-1) {2(q-1)\choose q-1}& 4^q\;\frac{1}{2}\;\sqrt{\frac{q}{`p}}\\ 
\hline
P^T_{2q+1,q-1} & \Tod(z)& %
\begin{array}{l}
  \frac{2z}{(1-4z)^2} +\\[2pt]
  \frac{z^2+z^3}{(1-4z)^3\sqrt{1-4z}}
\end{array}
&
\begin{array}{l}
q\;2^{2q-1} +\frac{q(q-1)(q-2)}{120}{2q\choose q}+\\[2pt]
\frac{(q+1)q(q-1)}{120}{2(q+1) \choose q+1}\\[2pt]
\end{array}
& 4^q\frac{1}{24}\;\sqrt{\frac{q^5}{`p}}
\\
\hline\hline
P^T_{2q,q}&\Ev(z) &  \frac{4z-1 + \sqrt{1-4z}}{2(1-4z)} & { 2q-1 \choose
  q}&4^q\;\frac{1}{2}\;\sqrt{\frac{1}{`p q}}\\[2pt]  \hline
P^T_{2q,q-1}&\Sev(z) &
\begin{array}{l}
\frac{z}{1-4z} +\\ \frac{z^2}{(1-4z)^2\sqrt{1-4z}}
\end{array}
& \begin{array}{l}4^{q-1} +\\ \frac{2 (2q - 5)  (2q - 3)  (2q
  -1)}{3(q-2)}  {2(q-3) \choose q-3}
\end{array}
& 4^q\;\frac{1}{12}\;\sqrt{\frac{q^3}{`p}}\\\hline
\end{array}
\end{displaymath}

and for \textit{normal forms}
\begin{displaymath}
\begin{array}{|l|l|l|l|l|}
\hline\hline
\textsl{coefficients}&\multicolumn{2}{c|}{\textsl{generating functions}}& \textsl{values} & \textsl{equivalents}\\
\hline\hline
P_{2q+1,q+1}^{NF} &  \Fod(z) & \frac{1-\sqrt{1-4z}}{2z} & C_q&4^q\sqrt{\frac{1}{\pi q^3}}\\
\hline
P_{2q+1,q}^{NF} &  \SFod(z) &  \frac{z}{\sqrt{1-4z}} + \frac{z^2}{(1-4z)\sqrt{1-4z}}& (q+1) {2q-3 \choose q-2} 
&4^q \frac{1}{8} \sqrt{\frac{q}{`p}}\\
\hline\hline
P_{2q,q}^{NF} &  \Fev(z) & \frac{z}{\sqrt{1-4z}} & 2{2q-3 \choose q-2} &4^q \frac{1}{4} \sqrt{\frac{1}{`p q}}\\
\hline
P_{2q,q-1}^{NF} & \SFev(z) &
\begin{array}{l}
\frac{z}{2\sqrt{1-4z}} + \frac{z}{2(1-4z)} +\\  \frac{2z^3 }{(1-4z)\sqrt{1-4z}} +\\
\frac{z^2}{\sqrt{1-4z}}+\\\frac{2z^4}{(1-4z)^2\sqrt{1-4z}}
\end{array} &
\begin{array}{l}
{2q-3 \choose q-2} + 2^{2q-3} + \\ (q-2){2q-2 \choose q-2} + \\
   2{2q-5 \choose q-3} + \\ \frac{(q-3)(q-2)}{3}{2q-5 \choose q-3}
\end{array} 
& 4^q\;\frac{1}{96}\;\sqrt{\frac{q^3}{`p}}\\
\hline
\end{array}
\end{displaymath}

We notice that the coefficients of the $P_n^{NF}$'s have the same asymptotic behavior as the
coefficients of $P_n^{T}$'s, with a slightly smaller constant, $1/8$ or $1/4$ for $1/2$
and $1/96$ for $1/12$.
Notice, in particular, that the results $P_{2q,q}^{NF} \sim \frac{1}{2}\, P_{2q,q}^{T}$
and $P_{2q+1,q}^{NF} \sim \frac{1}{4}\, P_{2q+1,q}^{T}$ comes from the identities.
\begin{eqnarray*}
  2{2q-3 \choose q-2} &=& \frac{q}{2q-1} {2q-1 \choose q}  \\
 (q+1) {2q-3 \choose q-2} &=& \frac{q+1}{2(2q-1)}\,(2q-1) {2(q-1)\choose q-1}.
\end{eqnarray*}

\section{Generating functions for terms}
\label{sec:gener-funct-terms}

We consider the vertical generating functions which gives the $T_{n,m}$'s for each value
of $m$.





\subsection*{Vertical generating functions}
\label{sec:vert-gener-funct}

We see that 
\[T_{n,m+1} = T_{n+1,m} - \sum_{k=0}^n T_{n-k,m} T_{k,m}.\]

Hence
\begin{eqnarray*}
  T^{\langle m\rangle} (0) &=& 0
\end{eqnarray*}
and
\begin{eqnarray*}
  T^{\langle m+1\rangle} (z) &=& \sum_{n=0}^{\infty} T_{n,m+1}  z^n\\
&=& \sum_{n=0}^{\infty} T_{n+1,m}  z^n - \sum_{n=0}^{\infty} \sum_{k=0}^n T_{n-k,m} T_{k,m} z^n\\
&=& \frac{T^{\langle m\rangle}(z)}{z} - (T^{\langle m\rangle}(z))^2.
\end{eqnarray*}
In other words
\[ z (T^{\langle m\rangle}(z))^2 - T^{\langle m\rangle}(z) + z T^{\langle m+1\rangle}(z) =0
.\]
Hence 
\[T^{\langle m\rangle}(z) = \frac{1 - \sqrt{1-4z^2T^{\langle m+1\rangle}(z)}}{2z}.\]
Moreover
\[[z]T^{\langle m\rangle}(z) = \frac{d \,T^{\langle m\rangle}}{d\,z}(0) = m.\]

We see that $T^{\langle m\rangle}$ is defined from $T^{\langle m+1\rangle}$.  Like the
bivariate generating function $T(z,u)$, $T^{\langle m\rangle}(z)$  is also difficult to
study, because we have $T^{\langle m\rangle}$ defined in term of 
$T^{\langle m+1\rangle}$.

\section{Conclusion}
\label{sec:conclusion}

We have given several parameters on numbers of untyped lambda terms and untyped normal
forms and proved or conjectured facts about them.  On another direction, it could be worth
to study typed lambda terms, whereas we have only analyzed untyped lambda terms in this
paper.

\nocite{bruijn58:_asymp_method_analy}

\begin{landscape}
  \begin{figure}[tb]
    \begin{center}
      \begin{math}
        \begin{array}[h]{r || r | r | r | r | r | r| r| r| r}
          n\backslash m& 0& 1 & 2 & 3 & 4 & 5 & 6 &  7 & \\
          \hline \hline
          1 & 0 & 1 & 2 & 3 & 4 & 5 & 6 & 7\\
          \hline
          2 & 1&2 & 3 & 4 & 5 & 6 & 7 & 8 \\
          \hline
          3 & 2& 4 & 8  & 14 & 22 & 32 & 44 & 58   \\
          \hline
          4 & 4& 12 & 26 & 46 & 72 & 104 & 142 & 186 \\
          \hline
          5 & 13& 38 & 87 & 172 & 305 & 498 &  763 & 1112 \\
          \hline
          6 & 42 & 127 & 324&  693 &  1294  & 2187 & 3432 &5089\\
          \hline
          7 & 139& 464& 1261& 2890 & 5831& 10684&18169 &29126\\
          \hline
          8 & 506& 1763& 5124& 12653 & 27254 & 52671 & 93488 &155129\\
          \hline
          9 & 1915& 7008& 21709& 57070& 130863& 269260& 508513& 896634\\
          \hline
          10 &   7558& 29019& 94840& 265129& 646458& 1406983& 2791564& 5136885 \\
          \hline
          11 & 31092& 124112& 427302& 1264362& 3262352& 7502892& 15703602 &30429782\\
          \hline
          12 & 132170& 548264& 1977908& 6168242& 16811366& 40776020& 89671904& 181746638\\
          \hline
          13 & 580466& 2491977& 9384672& 30755015& 88253310& 225197061& 520076012& 1104714147\\
          \hline
          14 & 2624545 & 11629836 & 45585471 & 156409882 & 471315501 & 1263116040 &3058077451
          & 6789961206 \\
          \hline
          15 & 12190623 &  55647539 &  226272369 &  810506769 &  2558249963 & 
          7184911623 &  18208806189 &  42244969589 \\
          \hline
          16 & 58083923 &  272486289 &  1146515237 &  4275219191 &  14098296495 &  41417170373
          &109721440529 &  265618096347 \\
          \hline
          17  & 283346273& 1363838742& 5923639803& 22933607180& 78832280277&
          241776779298& 668513708207& 1686996660888\\
          \hline
          18 & 1413449148& 6968881025& 31177380822& 125027527671& 
          446961983408&1428444131853&4116538065930& 10816530842627\\
          \hline
        \end{array}
      \end{math}
      \caption{Values of $T_{n,m}$ up to $(18,7)$}\label{fig:T}
    \end{center} 
\end{figure}
\end{landscape}

\begin{landscape}
  \begin{figure}[h]
    \centering
   \begin{math}
    \begin{array}[h]{ r | c | }
      n & P^T_n(m)\\
      \hline\hline
      1 & m \\
      \hline
      2 & m +1 \\
      \hline
      3 & m ^2 + m +2 \\
      \hline
      4 &3 m^{2}+5 m+4  \\
      \hline
      5 & 2m^3 + 6 m^2 + 17m +13\\
      \hline
      6 & 10m^3+ 26 m^2 + 49 m + 42\\
      \hline
      7 & 5 m^4 + 30 m^3 + 111 m^2 + 179m +139\\
      \hline
      8  & 35  m^4 + 134 m^3 + 405 m^2 + 683 m + 506\\
      \hline
      9 & 14 m^{5}+140 m^{4}+652 m^{3}+1658 m^{2}+2629 m+1915\\
      \hline
      10 & 126 m^{5}+676 m^{4}+2812 m^{3}+7122 m^{2}+10725 m+7558\\
      \hline
      11& 42 m^{6}+630 m^{5}+3610 m^{4}+12760 m^{3}+30783 m^{2}+45195 m+31092\\
      \hline
      12 & 462 m^{6}+3334 m^{5}+17670 m^{4}+60240 m^{3}+138033 m^{2}+196355 m+132170\\
      \hline
      13 & 132 m^{7}+2772 m^{6}+19218 m^{5}+87850 m^{4}+285982 m^{3}+635178 m^{2}+880379 m+580466\\
      \hline
      14 & 1716 m^{7}+16108 m^{6}+104034 m^{5}+449290 m^{4}+1390246 m^{3}+2991438 m^ {2}+4052459
      m+2624545\\
      \hline
      15 & 429 m^{8}+12012 m^{7}+99386 m^{6}+560854 m^{5}+2308173 m^{4}+6895122 m^{3}+14436365 m^{2}+19144575 m+12190623\\
      \hline
      16 & 6435 m^{8}+76444 m^{7}+584878 m^{6}+3076878 m^{5}+12039895 m^{4}+34815210
      m^{3}+71170791 m^{2}+92631835 m+58083923\\
      \hline
      17 & 1430 m^{9}+51480 m^{8}+502384 m^{7}+3389148 m^{6}+16925916 m^{5}+63753310 m^{4}+179178860 m^{3}+358339416 m^{2}+458350525 m+283346273\\
      \hline
      18 & 24310 m^{9}+357256 m^{8}+3176112 m^{7}+19799164 m^{6}+93981244 m^{5}+342274990 m^{4}+938333964 m^{3}+1840448776 m^{2}+2317036061 m+1413449148
    \end{array}
  \end{math}
    \caption{The polynomials $P^T_n$ for the function $m"|->" T_{n,m}$}\label{fig:Pn}
    \label{fig:poly}
  \end{figure}
\end{landscape}

\begin{landscape}
  \begin{figure}
    \centering
    \begin{math}
      \begin{array}[h]{r || r | r | r| r | r | r| r | r | r| r }
           n\backslash m& 0& 1 & 2 & 3 & 4 & 5 & 6 &  7 & 8 &\\
          \hline \hline
                1 &0 & 1 & 2 & 3 & 4 & 5 & 6 & 7 & 8 \\
\hline
        2 &1 & 2 & 3 & 4 & 5 & 6 & 7 & 8 & 9 \\
\hline
        3 &2 & 4 & 8 & 14 & 22 & 32 & 44 & 58 & 74 \\
\hline
        4 &4 & 10 & 20 & 34 & 52 & 74 & 100 & 130 & 164 \\
\hline
        5 &10 & 25 & 58 & 121 & 226 & 385 & 610 & 913 & 1306 \\
\hline
        6 &25 & 72 & 185 & 400 & 753 & 1280 & 2017 & 3000 & 4265 \\
\hline
        7 &72 & 223 & 614 & 1497 & 3244 & 6347 & 11418 & 19189 & 30512 \\
\hline
        8 &223 & 728 & 2195 & 5716 & 12863 & 25688 & 46723 & 78980 & 125951 \\
\hline
        9 &728 & 2549 & 8108 & 22745 & 56360 & 125093 & 253004 & 473753 & 832280 \\
\hline
        10 &2549 & 9254 & 31253 & 93734 & 244997 & 564854 & 1173029 & 2237558 & 3983189 \\
\hline
        11 & 9254 & 35168 & 124778 & 395720 & 1109222 & 2770904 & 6261818 & 12999728 & 25130630 \\
\hline
        12 & 35168 & 138606 & 512898 & 1720040 & 5097660 & 13347978 & 31308206 & 66902388 &
        132274680 \\
\hline
        13 & 138606 & 563907 & 2174894 & 7645095 & 23948550 & 66818531 &
        167837142 & 384821079 & 816168830 \\
\hline
        14 & 563907 & 2369982 & 9459993 & 34771380 & 114618495 & 335857722 &
        880524117 & 2092596528 & 4571548155 \\
\hline
        15 & 2369982 & 10231830 & 42221886 & 161568762 & 558056526 & 1723895502 &
        4785906510 & 12073186866 & 28016723742 \\
        \hline
        16 & 10231830 & 45381558 & 192944940 & 765787548 & 2764390146 &
        8947158690 & 25962816408 & 68135021640 & 163627733358 \\
        \hline
17 & 45381558 & 206266797 & 901441688 & 3701763855 & 13912595562 &
 47127027713 & 143678500332 & 397091138883 & 1005324501470\\
\hline
18 & 206266797& 959283300& 4302919895& 18223902654& 71123969121&
 251343711032& 799893538635& 2302171013970& 6046781201429\\
\hline
     \end{array}
    \end{math}
    \caption{Values of $F_{n,m}$ up to $(18,8)$}
    \label{fig:F}
  \end{figure}
\end{landscape}

\begin{landscape}
  \begin{figure}[h]
    \centering
   \begin{math}
    \begin{array}[h]{ r | c | }
      n & P_n^{NF}(m)\\
      \hline\hline
      1 & m \\
      \hline
      2 & m +1 \\
      \hline
      3 & m ^2 + m +2 \\
      \hline
4 & 2 m^{2}+4 m+4\\
\hline
5 & 2 m^{3}+3 m^{2}+10 m+10 \\
\hline
6 & 6 m^{3}+15 m^{2}+26 m+25\\
\hline
7 & 5 m^{4}+12 m^{3}+49 m^{2}+85 m+72\\
\hline
8 & 20 m^{4}+62 m^{3}+155 m^{2}+268 m+223\\
\hline
9 & 14 m^{5}+50 m^{4}+240 m^{3}+589 m^{2}+928 m+728\\
\hline
10 & 70 m^{5}+263 m^{4}+870 m^{3}+2146 m^{2}+3356 m+2549\\
\hline
11 & 42 m^{6}+210 m^{5}+1153 m^{4}+3658 m^{3}+8351 m^{2}+12500
m+9254\\
\hline
12 & 252 m^{6}+1128 m^{5}+4658 m^{4}+14838 m^{3}+33575 m^{2}+48987
m+35168 \\
\hline
13 & 132 m^{7}+882 m^{6}+5446 m^{5}+21198 m^{4}+63138 m^{3}+137695
m^{2}+196810 m+138606 \\
\hline
   14 & 924 m^{7}+4862 m^{6}+24086 m^{5}+93748 m^{4}+275898 m^{3}+587814
m^{2}+818743 m+563907\\
\hline
15 & 429 m^{8}+3696 m^{7}+25372 m^{6}+117120 m^{5}+429435
m^{4}+1223102 m^{3}+2558090 m^{2}+3504604 m+2369982\\
\hline
16 & 3432 m^{8}+20996 m^{7}+121286 m^{6}+556920 m^{5}+2011411 m^{4}+5601948 m^{3}+11448828 m^{2}+15384907 m+10231830\\
\hline
17 & 1430 m^{9}+15444 m^{8}+116892 m^{7}+624768 m^{6}+2717670
m^{5}+9524196 m^{4}+26064412 m^{3}+52459126 m^{2}+69361301
m+45381558\\
\hline
18 & 12870 m^{9}+90683 m^{8}+598120 m^{7}+3162562 m^{6}+13513606
m^{5}+46329205 m^{4}+124109404 m^{3}+245453736 m^{2}+319746317
m+206266797\\
\hline
    \end{array}
  \end{math}
    \caption{The polynomials $P^{NF}_n$ for the function $m"|->" F_{n,m}$}\label{fig:PFn}
    \label{fig:polyF}
  \end{figure}
\end{landscape}
\end{document}